\documentclass{article}

\usepackage{arxiv}

\usepackage[utf8]{inputenc} 
\usepackage[T1]{fontenc}    
\usepackage{hyperref}       
\usepackage{url}            
\usepackage{booktabs}       
\usepackage{amsfonts}       
\usepackage{nicefrac}       
\usepackage{microtype}      
\usepackage{lipsum}
\usepackage{fancyhdr}

\fancypagestyle{firstpage}
{
    \fancyhead[L]{}    
    \fancyhead[C]{ \scshape \scriptsize This work has been submitted to the IEEE for possible publication. Copyright may be transferred without notice, after which this version may no longer be accessible.}
}
\usepackage{graphicx}          
\usepackage{amsmath} 
\usepackage{amssymb}  

\usepackage{amsthm}
\usepackage{color,soul}
\usepackage{cite}
\usepackage{tabularx}
\usepackage{subcaption}

\allowdisplaybreaks
\newtheorem{thmm}{Theorem}

\newcommand*\diff{\mathop{}\!\mathrm{d}} 

\title{ Cyber-Attack Detection in Socio-Technical Transportation Systems Exploiting Redundancies Between Physical and Social Data}

\author{Tanushree~Roy,~Sara Sattarzadeh,~and~Satadru~Dey
\thanks{T. Roy, S. Sattarzadeh, and S. Dey are with the Department of Mechanical Engineering,
        The Pennsylvania State University, University Park, Pennsylvania 16802, USA.
        {\tt\small \{tbr5281,sfs6216,skd5685\}@psu.edu}.}%
}

\begin{document}
\maketitle
\thispagestyle{firstpage}
\begin{abstract}
Cyber-physical-social connectivity is a key element in Intelligent Transportation Systems (ITSs) due to the ever-increasing interaction between human users and technological systems. Such connectivity translates the ITSs into dynamical systems of socio-technical nature. Exploiting this socio-technical feature to our advantage, we propose a cyber-attack detection scheme for ITSs that focuses on cyber-attacks on freeway traffic infrastructure. The proposed scheme combines two parallel macroscopic traffic model-based Partial Differential Equation (PDE) filters whose output residuals are compared to make decision on attack occurrences. One of the filters utilizes physical (vehicle/infrastructure) sensor data as feedback whereas the other utilizes social data from human users' mobile devices as feedback. The Social Data-based Filter is aided by a fake data isolator and a social signal processor that translates the social information into usable feedback signals. Mathematical convergence properties are analyzed for the filters using Lyapunov's stability theory. Lastly, we validate our proposed scheme by presenting simulation results.
\end{abstract}


\section{Introduction}

{S}{ecurity} against cyber-attacks is one of the crucial criteria for today's emerging Intelligent Transportation Systems (ITSs). Such ITSs present both new threats as well as opportunities which were not existent in the conventional transportation systems. For example, while large-scale cyber-physical connectivity in ITSs exposes them to unprecedented vulnerability to cyber-threats, social connectivity through human-ITS interactions generate additional information for better management of such ITSs \cite{gowrishankar2014including,zheng2018framework}. Specifically, such social connectivity generates enormous amount of data that can be utilized meaningfully to enhance our system knowledge and transportation control algorithms. In this context, the goal of this work is to design safe-guarding mechanisms for ITSs using a system-theoretic framework and a data-fusion strategy.


\subsection{Background and motivation}
Data-fusion has been extensively used in the field of autonomous driving, energy, military surveillance and reconnaissance and medical fields\cite{becker2000sensor,grid_data_fusion,defense_data_fusion,book_data_fusion}. Based on the classification by Durrant-Whyte \cite{Data_fusion}, data fusion can be \textit{complementary, cooperative} or \textit{redundant}. A data fusion is called redundant if information from multiple data sources are aggregated for a single target system. This redundant data fusion technique is essential to increase the confidence in the data or to provide contingencies in case of corruption of one or more data sources. In relation to transportation, data is generated not only from vehicular and infrastructure (physical) sensors but also through (social) sensors such as mobile-devices generating position data through GPS or tweets. Data from these different modalities can be fused to gather traffic information. Specifically, we intend to exploit the \textit{redundancies} present in social and physical data. Such redundancy arises as physical infrastructure sensors measure vehicle positions whereas mobile device GPS or tweets also posses similar information. Having similar information from different data sources could be a key to detect anomalies in the traffic networks. In this work, we utilize such data fusion in conjunction with PDE system theoretic techniques for cyber-attack detection.

\subsection{Literature review}

Physical data in traffic infrastructure arises from vehicular and infrastructure sensors \cite{zheng2018framework,zheng2016big}. This data, in general, requires less processing and have higher reliability in terms of availability and delay. In contrast, social data \cite{Roy_its_2021} is generated by the human users-in-the-loop using mobile devices, such as messages on networking platforms and GPS coordinates of mobile devices. These social data are less-structured, less reliable in terms of availability and delay, and require more processing efforts. This distinction makes it especially challenging to combine these two types of data in a system-theoretic framework. 

With increasing interest in  detection of cyber-attacks on ITSs, many researchers have utilized physical data to detect Denial-of-Service (DoS) attack \cite{biron2017resilient,he2012mitigating,petrillo2018collaborative}, replay attack \cite{merco2018replay}, complex congestion pattern \cite{REILLY2016}, \cite{roy_dey2020}, attack on car platoons\cite{jahanshahi2018attack} and identity of attacker \cite{dadras2018identification}. On the other hand, most social data-based research focus on detecting traffic incidents or location using social data (such as Status Update Messages (SUM) or tweets) \cite{paule2019social,andrea2015social}. Very few works have utilized social data for real-time traffic control or estimation. For example, probe vehicle measurements have been used to estimate traffic density in \cite{herrera2010probe,Barreau2020probe}. However, utilization of social data to detect cyber-atacks in transportation systems has remained under-explored. Social data can be highly useful in modeling and real-time management of traffic systems. Especially, with the present rate of engagement in social networking, the amount of social data can be enormous compared to fewer physical data that are restricted by the cost of installation. Most importantly, in case of cyber-attacks in physical part of the infrastructure, social data would provide us real-time redundancies in monitoring traffic states. In \cite{canepa_2013}, spoofing attack detection strategies are proposed using probe-based traffic flow information. This work uses a discretized approximation of a first order traffic model to obtain its attack detection scheme using a mixed-integer linear programming. While \cite{canepa_2013} uses a macroscopic traffic model, it is only first order which fails to capture realistic human driving behavior \cite{daganzo1995requiem}. Moreover, the analysis is based on discretization in time domain which may further add to inaccuracy of the proposed strategy.


In our recent work \cite{Roy_its_2021}, we proposed a diagnostic framework that utilized microscopic platoon model and system theoretic tools along with physical and social data to detect cyber-attacks affecting vehicular communication network. The first limitation of \cite{Roy_its_2021} is the use of microscopic model that requires extensive computation and may be unrealistic in the context of numerous links in the traffic network or for capturing complex traffic dynamics pattern. Second, the cyber-attack considered was limited to that on vehicular communication networks and attack such as complex congestion pattern cannot be captured. Third, only Connected and Autonomous Vehicle (CAV) dynamics is considered while in reality cyber-attack or mis-information on traffic infrastructure can disrupt traffic flow irrespective of autonomous or human-driven or heterogeneous vehicles. Lastly, in our previous framework, no contingency for cyber-attack on social data has been considered.



\subsection{Main contribution}
 This paper simultaneously addresses the aforementioned limitations in our previous work and bridges the research gap in using physical and social data to detect infrastructure or global cyber-attack such as coordinated ramp metering, data spoofing on fixed sensors, adversarial or hijacked vehicle interventions. In other words, the main contribution of this work lies in \textit{design of macroscopic model-based cyber-attack detection scheme for traffic networks that utilizes redundancies between physical and social data.}

In this paper, we intend to detect infrastructure level cyber-attack using a second order continuum or macroscopic traffic model. This not only enables us to characterize global and complex behaviour of traffic, but also help simulate infrastructure level cyber-attacks without the computation burden of a vehicle-level microscopic model. Specifically, our proposed scheme contains two parallel macroscopic traffic model-based Partial Differential Equation (PDE) filters whose output residuals are compared to make decision on attack occurrences. One of the filters utilizes physical sensor data as feedback whereas the other utilizes social data from human users' mobile devices as feedback. The Social Data-based Filter is aided by a fake data isolator and a social signal processor that translates the social information into usable feedback signals. We analyze the convergence properties and design the filter parameters using Lyapunov's stability theory. 

In our previous work \cite{roy_dey2020}, we designed an attack detection filter combining second order macroscopic Aw-Rascle-Zhang (ARZ) traffic model and outlet traffic flux measurement from physical infrastructure. The major distinction between our current work and \cite{roy_dey2020} lies in the following: (i) the current work utilizes both physical and social data while \cite{roy_dey2020} uses only physical data; (ii) the current work has two filters working in parallel while \cite{roy_dey2020} uses only one filter; (iii) the current work utilizes ARZ model with traffic \textit{density and velocity} states while \cite{roy_dey2020} adopted ARZ model with \textit{flux and velocity} states. Such difference in the ARZ model translates to different coupling in the PDE models and slight difference in backstepping transformation used for filter design.

\subsection{Paper organization and notation}
The rest of the paper is organized as follows: Section II discusses the traffic modeling and problem set-up, Section III  discusses the design of the Social Data-based Filter while Section IV introduces the Physical Data-based Filter and Section V discusses the Comparator. Finally, Section VI presents the simulation results followed by conclusion in Section VII.

\textbf{Notation:} The following notations has been used in this work $w_t = {\partial w}/{\partial t}$, $w_x = {\partial w}/{\partial x}$; $\dot{w}(p) = \frac{\diff w}{\diff p}$, $\ddot{w}(p) = \frac{\diff^2 w}{\diff p^2}$; $\left\|w(.) \right\|$ denotes the spatial $\mathcal{L}_2$ norm given as $\left\|w(.) \right\| := \sqrt{\int_0^L w^2(x) dx}$.

\section{Traffic Modeling and Problem Statement}
In this section, we describe the problem set-up in terms of the traffic model, cyber-attack threats in socio-technical traffic system, physical and social data, and cyber-attack detection scheme.

\subsection{Macroscopic traffic model}
In this study, we consider traffic flow on a single lane of a free-way as a one-dimensional spatial evolution along $x\in [0,L]$, where $L$ represents the length of the free-way. In general, traffic flow is characterized by two of the three main variables: density, velocity and flux. In congruence with the data measurement available to our problem, we choose density and velocity as our variables. These variables are defined as follows: (i) \textit{traffic density }is the number of vehicles per unit length of the freeway, and (ii)\textit{ traffic velocity} is given by the mean distance covered per unit time. 

In this framework, traffic  flow is represented by Aw-Rascle-Zhang (ARZ) spacio-temporal macroscopic traffic model where the dynamics of the density $\rho(x,t)$ and velocity $v(x,t)$ of the traffic flow are given as follows \cite{AwRascle}:
\begin{align}\label{macro1}
    & \rho_t +(\rho v)_x = 0,\quad v_t + (v+\dot{V}_{opt}(\rho)\rho) v_x=\frac{V_{opt}(\rho)-v}{T_r},
\end{align}
where $t \in [0,\infty)$ is time and $x \in [0,L]$ is the spatial variable, $V_{opt}(\rho)$ is the optimal velocity function and $T_r$ is the relaxation parameter. The boundary conditions are given by 
\begin{align}\label{BC1}
 &\rho(0,t) = \rho_m + u,\quad \rho(L,t) = \rho^*,
\end{align}
where $\rho_m$ is the maximum density and $u$ is a control input at the inlet ramp metering on the free-way. Here $\rho^*$ is constant nominal density maintained at the outlet of the freeway using loop detector measurement. 

To facilitate our analysis, we first linearize the ARZ model \eqref{macro1} around nominal operating point $(\rho^*,v^*)$ to obtain 
\begin{align}\label{lin1}
    &\rho_t+ v^*\rho_x+\rho^*v_x =0,\\\label{lin2}
    &v_t+[v^*+\rho^*\dot{V}_{opt}(\rho^*)]v_x =\frac{\dot{V}_{opt}(\rho^*)\rho-v}{T_r}.
\end{align}
Subsequently, to decouple this linearized system, we perform the following transformation \cite{costeseque2015lax}
\begin{align}
    w:=v-\rho \dot{V}_{opt}(\rho^*),
\end{align}
to obtain the transformed model in $w-v$ domain:
\begin{align}\label{sys1}
    &w_t + k_1 w_x =-k_2w,\quad
    v_t + k_3 v_x = -k_2w,
\end{align}
where $k_1 = v^*, k_2=\frac{1}{T_r}$ and $k_3=[v^*+\rho^*\dot{V}_{opt}(\rho^*)]$. Under attack, this model is modified as
\begin{align}\label{sys11}
    &w_t + k_1 w_x =-k_2w,\quad v_t + k_3 v_x = -k_2w+\delta_1,\\ \label{bc1122}
    & w(0,t)= v(0,t) -(\rho_m+u+d_2-\rho^*)\dot{V}_{opt}(\rho^*),\\\label{bc11}
    & w(L,t) = v(L,t),
\end{align}
where $\delta_1(x,t)$ is the  in-domain distributed attack while $\delta(t)$ represents the inlet boundary attack. Without loss of generality, we can assume here $u=\rho^*-\rho_m$ and define $\delta_2:=-d_2/\dot{V}_{opt}(\rho^*)$ which modifies \eqref{bc1122} as
\begin{align}\label{bc22}
    & w(0,t)= v(0,t) +\delta_2.
\end{align}
We also note here that attacks at the outlet boundary are not considered as it is be trivial to analyze with our assumed outlet boundary measurement.

\subsection{Cyber-attack threats on socio-technical traffic system}

The traffic model, described in \eqref{sys11} and \eqref{bc11}-\eqref{bc22}, captures two broad categories of cyber-attacks: namely in-domain and boundary. The in-domain cyber-attack $\delta_1(x,t)$ is a distributed attack on the traffic system and models attacks on communication layers of Connected Autonomous Vehicles, spoofing of vehicle navigation systems or hijacked rogue vehicles. The in-domain attack can also model cyber-attacks on central management system since the latter can effect large scale traffic disruptions by compromising wide range of traffic management devices on the freeway. These attacks can also be on the network layer, spoofing data to and from the central management system \cite{REILLY2016}.

On the other hand, the inlet cyber-attack $\delta_2(t)$ models attacks on traffic control infrastructure such as ramp-metering or loop detectors. These attack can be orchestrated by hacking into the software of these devices/sensors/control boxes or the central management system. 

Additionally, in a  socio-technical traffic system, cyber-attacks may arise from corrupted social data. Such attacks can be initiated through generation of fake tweets or messages and spams. Utilization of social signals for traffic management system is thus endangered by such cyber-attacks. Furthermore, corrupted social data may include malware or phishing links that can jeopardise the operation of central management systems. These can also lead to false panic scenarios and severely impact traffic management systems.

\subsection{Physical and social data}
The measurement of the outlet traffic flow $v(L,t)$ is used as a physical measurement data from the infrastructure sensors. Concurrently, we also consider $N$ social-agents or human-in-the-loop, who are travelling in cars on this freeway and are transmitting social data such as geo-tags, GPS coordinates from portable devices, messages, Status Message Updates (SUMs) or tweets. When any car begins transmitting positional information through social data, it not only provides us with the position knowledge of these vehicles but also provides us an unique access to information of traffic at these in-domain points. 

Note that the social data contains vehicle-level (that is, microscopic) information whereas the model variables are global traffic quantities such as density, flow velocity and flux. In order to translate the microscopic data to macroscopic variables, we leverage the connection between the macroscopic ARZ model and microscopic car-following model \cite{AwRascle}. This motivates us to identify social-data generating segments of the macroscopic traffic as microscopic. Subsequently, we define the dynamics of each social-data transmitting vehicle using a modified car-following model \cite{AwRascle}
\begin{align}\label{micro}
    \ddot{z}_i(t) =& C_{\gamma}\frac{\dot{b}_{i}(t)-\dot{z}_i(t)}{(b_{i}(t)-z_i(t))^{\gamma+1}}+\frac{1}{T_r}\left[V_{opt}\left(\frac{\Delta X}{b_i(t)-z_i(t)}\right)-\Dot{z}_i(t)\right],
\end{align}
where $z_i, \forall i\in \{1,\hdots,N\}$ represent the position of the $i$-th car transmitting social data, $\Delta X$ is the length of the car and $b_i$ is the position of the micro-macro interface. The interface $b_i$ here acts analogous to  preceding car in the car-following model. The constant parameter $C_{\gamma}=v^*(\Delta X)^\gamma$ where $\gamma$ is a non-negative parameter.

Since we obtain positional information from the social data of the $i$-th social agents, we assume that the position data of their vehicle $z_i$ is known. The optimal velocity function $V_{opt}$ is monotonic and bounded. This implies that it is invertible. Thus, the position of the interface $b_i$ can be computed from \eqref{micro} by solving a nonlinear differential equation of the form $\dot{b}_i=h(b_i, z_i, \dot{z}_i, \ddot{z}_i).$
This enables us to define local density as
\begin{align}\label{micro2}
    \rho_i = \Delta X/(b_i-z_i)\approx \rho(z_i,t).
\end{align}
These $\rho(z_i,t)$ then serve as measurements for the Social Data-based Filter.

\subsection{Adversarial attack detection approach}

Based on the above setup, the cyber-attack detection scheme is shown in Fig. \ref{scheme}. The scheme contains \textit{Social Data-based Filter}, \textit{Physical Data-based Filter} and \textit{Comparator}. The \textit{Social Data-based Filter} utilizes $\rho(z_i,t)$ as feedback whereas the \textit{Physical Data-based Filter} utilizes $v(L,t)$ for employing output injection. Both of these filters output residual signals. These residual signals are processed by the \textit{Comparator} to produce a decision of attack detection.  In the following subsections, we will discuss the Filters and the Comparator.

\begin{figure}[tb]
      \centering
      \includegraphics[width=0.7\textwidth]{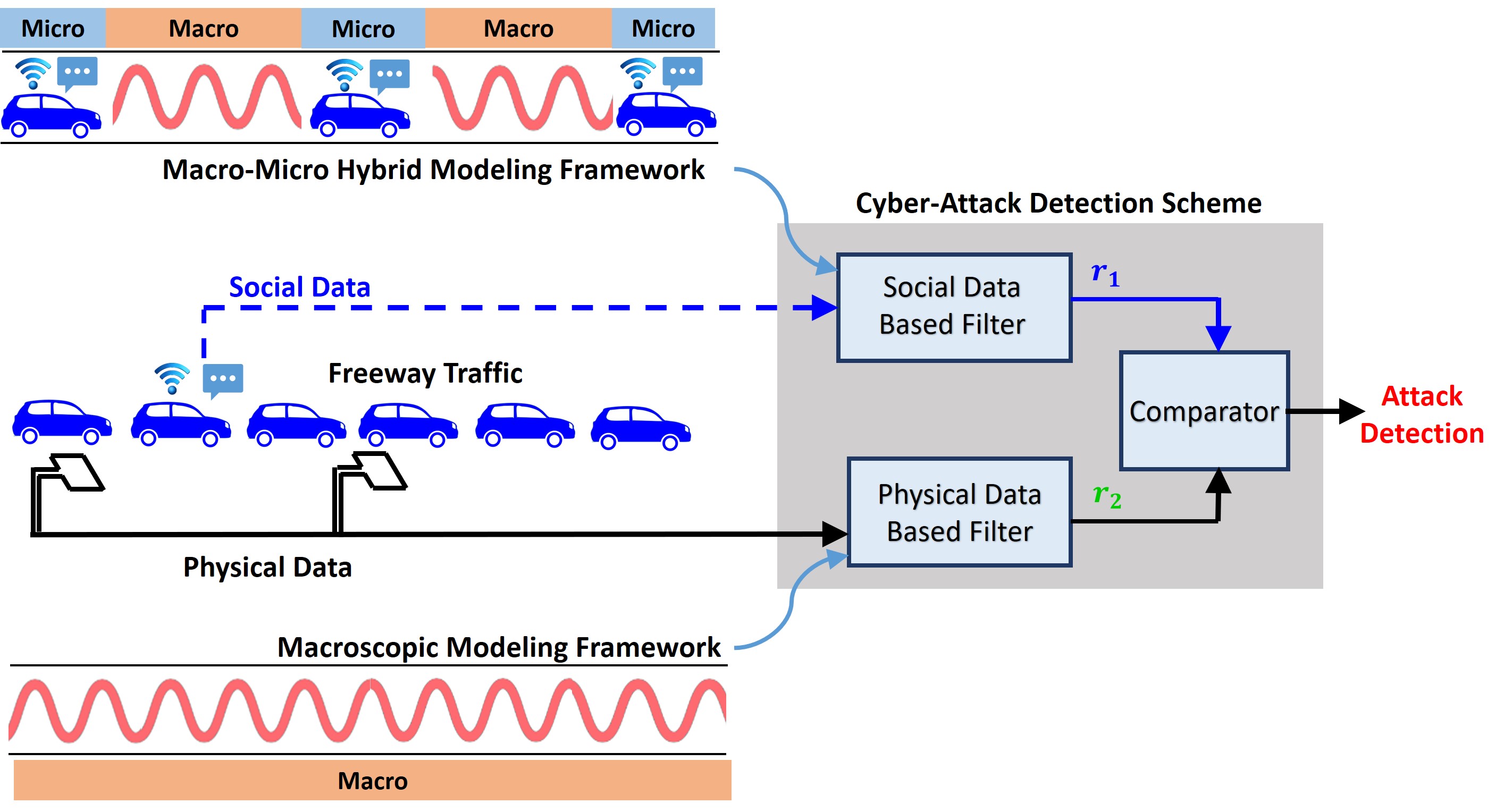}
      \caption{Cyber-attack detection scheme using social and physical data on distributed traffic model.}
      \label{scheme}
\end{figure}

\section{Social Data-based Filter}

Online Social Networking (OSN) has grown extensively over the recent years. Among various providers of OSN, Twitter alone provides networking platform to roughly one-fourth of US adult population of which 42\%  are on the platform daily\footnote{https://www.omnicoreagency.com/twitter-statistics/}. Furthermore, about 40\% of these Status Update Messages (SUMs) or tweets contain information such as geo-tags and geographical location information \cite{paule2019fine}. This suggests that the data from OSN can be leveraged successfully as social data signal in the context of traffic systems.

To achieve the aforementioned objective, we propose a Social Data-based Filter (see Fig. \ref{SDF}). This filter consists of three separate stages: Stage I is Fake Social Data Isolator (FSDI), Stage II is Social Signal Processor (SSP) and finally Stage III consists of Social Signal Residual Generator (SSRG). These three stages work in tandem to ensure high quality prediction from the Social Data-based Filter.

\begin{figure}[h!]
      \centering
      \includegraphics[scale=0.25]{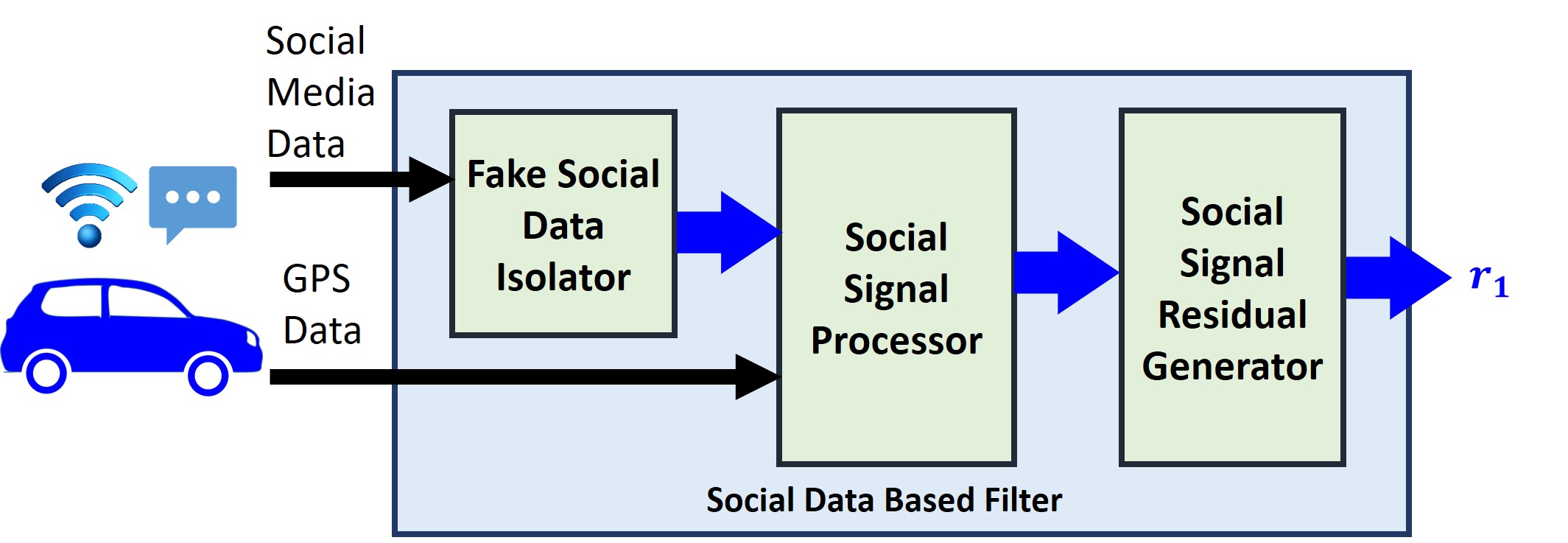}
      \caption{Social Data-based Filter.}
      \label{SDF}
\end{figure} 

The purpose of the three stages are of gate-keeping, pre-processing and social signal residual generation. Here the Stage I Fake Social Data Isolator detects cyber-attack on the social signals by classifying fake or illegitimate social data. This stage then passes legitimate signals onto Stage II. At this stage, the legitimate social signals undergo pre-processing in terms of acquiring meta-data from geo-tags or obtaining location information from known landmarks in the target area. Once this position information is received, it is send to Stage III. In case of GPS data, Stage II directly passes this to Stage III. Finally at Stage III, the Social Signal Residual Generator utilizes these position measurements to generate the social signal residual. This residual is then send to the Comparator block (along with physical signal residual) in order to obtain an attack decision.

\subsection{Fake Social Data Isolator (FSDI)}
 With vast social connectivity, Twitter and other OSN platforms have become a thriving ground for spamming accounts or tweets, bots, fake news, ad spam or phishing links, linkbait-and-switch, malware  or ``too-good-to-be-true" spams. For instance, Twitter systems flagged 3.2 million spamming accounts on their platform on 2018.  This signifies the importance of isolating fake tweets or spam from the credible social signals in order to retrieve meaningful traffic information from these social data. In other words, we have to ensure that the social data feedback is credible. 

In this work, we isolated credible social data from fake or spam tweets by implementing a Long-Short Term Memory (LSTM) recurrent neural network (RNN) \cite{LSTM}. Like any standard recurrent network, the LSTM-RNN has the ability to derive long-term dependency information from the input data. This enables the neural network model to capture contextual information from the social signal and effectively distinguish between fake and genuine ones. A standard LSTM recurrent neuron has the following gates: (i) input gate, (ii) candidate gate, (iii) forget gate and (iv) output gate (Fig. \ref{fig:LSTM}). 

The network designed for our problem has 2 hidden layers containing 48 and 24 neurons successively. We use sigmoid activation function to evaluate the classifier output. The input to LSTM network, if needed, is padded with zeros or truncated to standardize an input length and the size of the input layer for LSTM is equal to the size of the vocabulary obtained from the tokenizer.  We set our learning rate to be $0.01$ and binary cross-entropy as the loss function. A drop-out of $0.4$ and batch regularization are chosen to avoid over-fitting of the data.

\begin{figure}
    \centering
    \includegraphics[scale=0.25,width=0.5\textwidth]{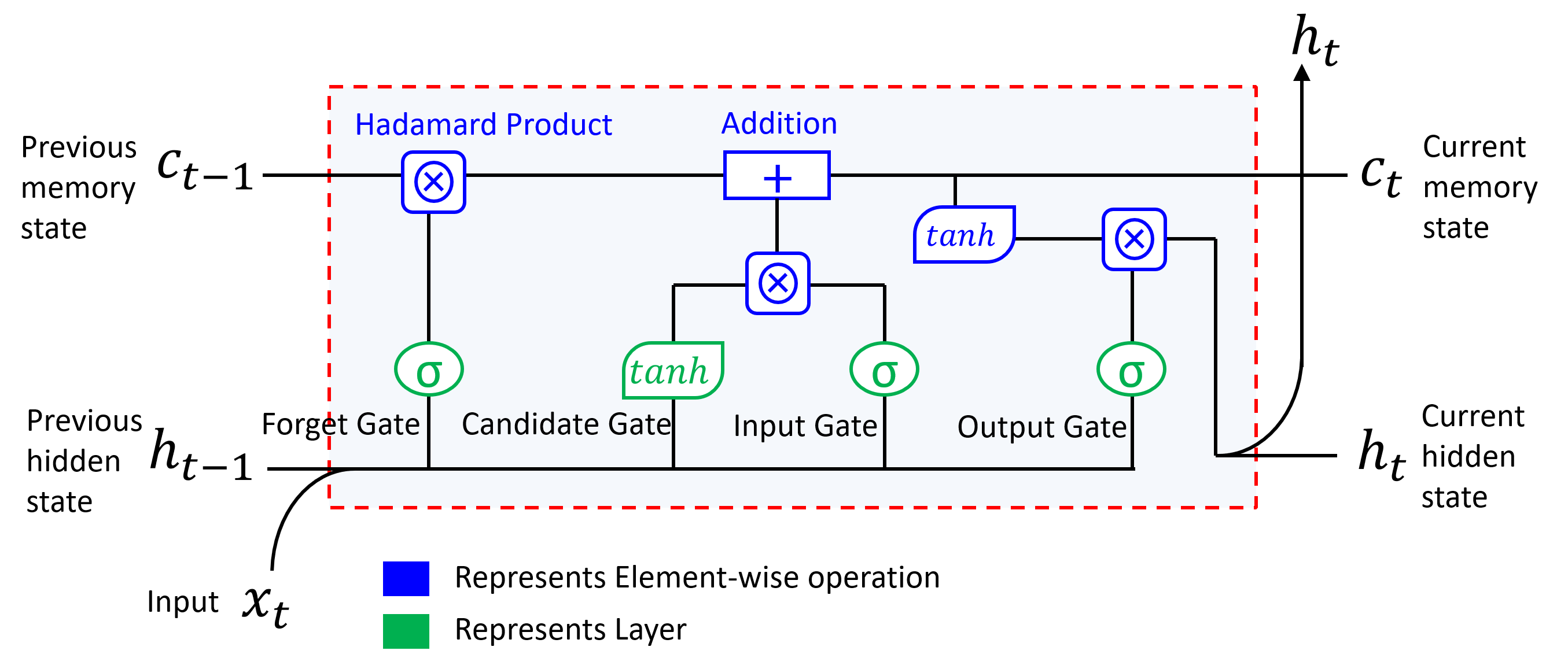}
    \caption{The structure of a single LSTM cell.}
    \label{fig:LSTM}
\end{figure}

In terms of data, we have generated set of artificial tweets. For our problem, real life fake data could not be gathered due to lack of available dataset. In order to show the proof-of-concept for the scheme, we generated synthetic data containing legitimate tweets resembling real tweets mentioning landmark locations. Furthermore, fake/corrupted data were generated by adding or manipulating these tweet data in the following way: (i) addition of spam links, (ii) inclusion of ``too-good-to-be-true" claims, (iii) mention of fake or out-of-scope landmarks and (iv) fake tweets (such as fake restaurant promotions, false donation promises in exchange of re-tweets etc). 
It should be noted here that classification of credible vs illegitimate tweets has been attempted in literature using content-based and account-based strategies\cite{Coulter_fake_tweets}. Here, we have restricted ourselves to content-based classification problem as the account-based study is beyond the scope of this work.
The isolated credible social data can then be utilized by the Social Signal Processor in Stage II.

\subsection{Social Signal Processor (SSP)}

We utilize the Social Signal Processor introduced in our previous work \cite{Roy_its_2021} to generate positions or trajectories information of vehicles in the traffic. The processor first collects social data in the form of meta-data such as geo-tags or location data from tweets as well as GPS signals from the mobile devices (when turned ON). Thereafter, position information from GPS signals or geo-tags are passed on to SSRG as is. Simultaneously, this processor also ensures that the collected tweets are \textit{relevant}. The \textit{relevant} texts  are defined as those which provide information about the \textit{present} location of a user (connected to the vehicle) while \textit{irrelevant} texts do not. For example, among the following tweets -- `` I'm near Pho 11 Vietnamese Restaurant", ``\#Beef \#Pho @Pho 11 Vietnamese Restaurant is terrific!" and `` Last night party at pho 11 was a blast" -- only the first is \textit{relevant} as it can be used to extract present position data of the vehicle. This filtering is achieved using standard Natural Language Processing (NLP) techniques \cite{Roy_its_2021}. 

Thus, using the knowledge of position of landmarks along the traffic of interest and the GPS or geo-tags signals of the social data from the vehicles, the positions or trajectories of the vehicles can be successfully obtained from the SSP. Once position is obtained, it is used to estimate the density of the traffic in those positions or along those trajectories.  This estimates are eventually passed as a feedback to the Social Signal Residual Generator  in Stage III.

\subsection{Social Signal Residual Generator (SSRG)}

Position measurement from the SSP is subsequently used to generate the local density using \eqref{micro}-\eqref{micro2}. This density measurements is used by the Social Signal Residual Generator (SSRG) and are given by
\begin{align}
&y_i(t) = \rho(z_i,t), \forall i={1,2,\hdots,N},
\end{align}
where $N$ represents the number of social-agents or human-in-the-loop transmitting social data. In transformed variables, this can be expressed as
\begin{align}
&y_i(t)= \frac{v(z_i,t)-w(z_i,t)}{\dot{V}_{opt}(\rho^*)}, \forall i={1,2,\hdots,N}.
\end{align}
These $N$ output can also be represented as a vector in the following form \cite{Wang_distributed}:
\begin{align}
    y_s = \int_0^L c(x)\frac{{v}(x,t)-{w}(x,t)}{\dot{V}_{opt}(\rho^*)}\diff x,
\end{align}
where $c(x) = [c_1(x), c_2(x), \hdots, c_N(x)]^T$ and $c_i(x)$ is defined in terms of dirac delta functions $c_i(x)= \delta(x-z_i), \forall i={1,2,\hdots,N}$. Next, the dynamics of the model-based SSRG is given by:
\begin{align}\label{social1}
&\widetilde{w}_t=-k_1\widetilde{w}_x-k_2{\widetilde{w}}+ \alpha \dot{V}_{opt}(\rho^*) c(x)^T(y_s-\widetilde{y}_s),\\\label{social2}
&\widetilde{v}_t=-k_2\widetilde{v}_x-k_2{\widetilde{w}}+\beta \dot{V}_{opt}(\rho^*) c(x)^T(y_s-\widetilde{y}_s),\\\label{social3}
&\widetilde{w}(0,t)= \widetilde{v}(0,t)=v(0,t) ,\quad  \widetilde{w}(L,t) = \widetilde{v}(L,t),\\
&\widetilde{y}_s=\int_0^L c(x)\frac{\widetilde{v}(x,t)-\widetilde{w}(x,t)}{\dot{V}_{opt}(\rho^*)}\diff x , \forall i={1,2,\hdots,N},
\end{align}
where $\widetilde{w}, \widetilde{v}$ and $\widetilde{y}_s$ represent the estimates of $w, v$ and $y_s$, respectively; $\alpha = \text{diag}(\alpha_1,\alpha_2, \hdots, \alpha_N)$,  $\beta = \text{diag}(\beta_1,\beta_2, \hdots, \beta_N)$ and $\alpha_i,\beta_i, \forall i$ are the filter gains. Here we assume that the inlet velocity $v(0,t)$ is measured and known and is thus used as the boundary condition for the filter as well. 

Thus, the errors for the SSRG can be defined by  ${\widetilde{W}:=w-\widetilde{w}}$ and $\widetilde{V}:=v-\widetilde{v}$ and their dynamics is given by the difference of \eqref{social1}-\eqref{social3} and \eqref{sys11},\eqref{bc11}-\eqref{bc22}:
\begin{align}\label{soc_err1}
&\widetilde{W}_t=-k_1\widetilde{W}_x-k_2{\widetilde{W}}-\alpha \dot{V}_{opt}(\rho^*)c(x)^T(y_s-\widetilde{y}_s),\\
&\widetilde{V}_t=-k_3\widetilde{V}_x-k_2{\widetilde{W}}-\beta \dot{V}_{opt}(\rho^*)c(x)^T(y_s-\widetilde{y}_s)+\delta_1,\\\label{soc_err2}
&\widetilde{W}(0,t)= \delta_2,\quad  \widetilde{W}(L,t) = \widetilde{V}(L,t).
\end{align}
The residual for the Social Data-based Filter is given by
\begin{align}\label{residual_social}
& r_s(t) = \sum_{i} \rho^2(z_i,t)\dot{V}^2_{opt}(\rho^*) = \sum_{i} \left(\widetilde{V}(z_i,t)-\widetilde{W}(z_i,t)\right)^2,
\end{align}
where the dynamics of $z_i$ is given by \eqref{micro}. With this we present our first theorem that provide the design conditions for filter gains $\alpha$ and $\beta$ of the Social Data-based Filter.

\begin{thmm}[Convergence of Residual Dynamics for Social Data-based Filter]
Consider the error dynamics \eqref{soc_err1}-\eqref{soc_err2}, residual definition \eqref{residual_social} and the dynamics of social data sensors \eqref{micro}. The residual signal is asymptotically stable in the following sense:
\begin{align}\label{thm11}
    & r_s(t)\leqslant c_1 r_s(0)\exp{(-\lambda_s t)}, 
\end{align}
 for some $c_1>0$ and $\lambda_s >0$ without any attack. Furthermore, the residual signal is input-to-state state stable in the following sense:
 \begin{align}\label{thm12}
    & r_s(t)\leqslant c_1 r_s(0)\exp{(-\lambda_s t)}+c_2\left( \|\delta_1(.,t)\|^2+\delta_2^2(t)\right), 
\end{align}
 for some $c_2>0$ under attack. These conditions can be guaranteed if there exists a positive scalar $\lambda_s$ such that LMI condition $(\mathcal{P}+\lambda_s I)<0$ is satisfied and matrix $\mathcal{P}$ is such that 
 \begin{align}
     \overline{\sigma}(\mathcal{P}_i)\leqslant \overline{\sigma}(\mathcal{P}), \, \forall i
 \end{align}
 where $\overline{\sigma}$ represents the maximum singular value and matrix $\mathcal{P}_i$ is given as follows:
    \begin{align}\label{P_mat}
    &\mathcal{P}_i=\frac{1}{\Tilde{x}_i-\Tilde{x}_{i+1}}\begin{bmatrix}
     \alpha_i &-\frac{1}{2}(\alpha_i-\beta_i) \\
    -\frac{1}{2}(\alpha_i-\beta_i) & -\beta_i
    \end{bmatrix},
\end{align}
where for all $i=1, \hdots, N$, the interval $(\Tilde{x}_{i+1}, \Tilde{x}_{i})$ contains the position $z_i$ of the $i$th vehicle. 
\end{thmm}

\begin{proof}
Consider the following Lyapunov candidate functional:
\begin{align}
    &\mathcal{E}(t)=\mathcal{E}_1(t)+\mathcal{E}_2(t),
    \end{align}
    where
    $\mathcal{E}_1{=}\int_{0}^L \frac{e^{-x}}{2}\widetilde{W}^2(x,.)\diff x$ and $\mathcal{E}_2{=}\int_{0}^L \frac{e^{-x}}{2}\widetilde{V}^2(x,.)\diff x$.
Taking derivative of $\mathcal{E}_1(t)$ and $\mathcal{E}_2(t)$ with respect to time and substituting \eqref{soc_err1}-\eqref{soc_err2} we obtain:
\begin{align}
    \dot{\mathcal{E}}_1(t)=&  -\frac{k_1}{2}[e^{-L}\widetilde{W}^2(L)-\widetilde{W}^2(0)]-\left(k_2+\frac{k_1}{2}\right)\int_0^Le^{-x}\widetilde{W}^2-\sum\limits_{i=1}^N\alpha_ie^{-z_i}\left[-\widetilde{W}^2(z_i)+\widetilde{V}(z_i)\widetilde{W}(z_i)\right] ,\\
    \dot{\mathcal{E}}_2(t)=&  -\frac{k_3}{2}[e^{-L}\widetilde{V}^2(L)]-\frac{k_3}{2}\int_0^Le^{-x}\widetilde{V}^2-\frac{k_2}{2}\int_0^L\!\!\!e^{-x}\widetilde{W}\widetilde{V} -\sum\limits_{i=1}^N\beta_ie^{-z_i}\left[\widetilde{V}^2(z_i)-\widetilde{V}(z_i)\widetilde{W}(z_i)\right]+\int_0^Le^{-x}\widetilde{V}\delta_1.
\end{align}
Now using mean value theorem, we can assert that for each position $z_i$ of the $i$th vehicle we can find an interval  $[\Tilde{x}_{i+1},\Tilde{x}_i]\subset [0,L], \forall i \in \{1,\hdots,N\}$ such that  \cite{collocated_lyapunov_proof}
\begin{align}\label{mvt}
    &e^{-z_i}\left[\alpha_i\widetilde{W}^2(z_i)-\beta_i\widetilde{V}^2(z_i)+(\beta_i-\alpha_i)\widetilde{V}(z_i)\widetilde{W}(z_i)\right]
    =\sum\limits_{i=1}^N\frac{1}{\Tilde{x}_i-\Tilde{x}_{i+1}}\int\limits^{\Tilde{x}_i}_{\Tilde{x}_{i+1}}e^{-x}\left[\alpha_i\widetilde{W}^2-\beta_i\widetilde{V}^2+(\beta_i-\alpha_i)\widetilde{V}\widetilde{W}\right].
\end{align}
We note here that $0<\Tilde{x}_{N+1}<\hdots<\Tilde{x}_1<L$. Let us define two more points $\Tilde{x}_{N+2}:=0$ and $\Tilde{x}_0=L$. This implies any integral over the domain $[0,L]$ can be written as follows:
\begin{align}\label{sum_int}
    \int\limits_0^L f(x)\diff x = \sum\limits_{i=0}^{N+1} \int\limits^{\Tilde{x}_i}_{\Tilde{x}_{i+1}}f(x)\diff x.
\end{align}
Next let us define a positive constant $\gamma$ such that $0<\gamma<k_3$. Then, using \eqref{mvt}, \eqref{sum_int} and Young's Inequality in $ \dot{\mathcal{E}}_1(t)$ and $ \dot{\mathcal{E}}_2(t)$ yields
\begin{align}
     \dot{\mathcal{E}}(t)\leqslant &\frac{k_1}{2}\widetilde{W}^2(0)-\sum\limits_{i=0}^{N+1} \int\limits^{\Tilde{x}_i}_{\Tilde{x}_{i+1}}e^{-x}\Bigg[\left(k_2+\frac{k_1}{2}\right)\widetilde{W}^2
     +\left(\frac{k_3-\gamma}{2}\right)\widetilde{V}^2-k_2\widetilde{W}\widetilde{V}\Bigg]\diff x\nonumber\\ &+\sum\limits_{i=1}^N\frac{1}{\Tilde{x}_i-\Tilde{x}_{i+1}}\int\limits^{\Tilde{x_i}}_{\Tilde{x}_{i+1}}e^{-x}\bigg[\alpha_i\widetilde{W}^2-\beta_i\widetilde{V}^2+(\beta_i-\alpha_i)\widetilde{V}\widetilde{W}\bigg]\diff x+\frac{1}{2\gamma}\int\limits_0^L\delta^2_1(x,t)\diff x.
\end{align}
We also note here that $\widetilde{W}(0)=\delta_2$. Next we define a vector $\zeta = e^{-x/2}[\widetilde{W}, \widetilde{V}]^T$, matrix $\mathcal{P}_i$ as in \eqref{P_mat}, and a negative definite matrix $\mathcal{Q}$ as follows
\begin{align}
    \mathcal{Q}:= \begin{bmatrix}
    -\left(\frac{k_1}{2}+k_2\right) & -\frac{k_2}{2}\\
    -\frac{k_2}{2} &-\frac{k_3-\gamma}{2}
    \end{bmatrix}<0.
\end{align}
This implies that 
\begin{align}
     \dot{\mathcal{E}}(t)\leqslant & \sum\limits_{i=1}^N\int\limits^{\Tilde{x}_i}_{\Tilde{x}_{i+1}} \zeta^T(\mathcal{P}_i+\mathcal{Q})\zeta + \sum\limits_{i=0,N+1}\int\limits^{\Tilde{x}_i}_{\Tilde{x}_{i+1}}\zeta^T\mathcal{Q}\zeta+\frac{1}{2\gamma}\|\delta_1(.,t)\|^2+\frac{k_1}{2}\delta_2^2.
\end{align}
As $\mathcal{Q}$ is negative definite, $\mathcal{P}_i+\mathcal{Q}<\mathcal{P}_i$ which yields
\begin{align}
     \dot{\mathcal{E}}(t)\leqslant & \sum\limits_{i=0}^{N+1}\int\limits^{\Tilde{x}_i}_{\Tilde{x}_{i+1}}\zeta^T\mathcal{P}_i\zeta+\frac{1}{2\gamma}\|\delta_1(.,t)\|^2+\frac{k_1}{2}\delta_2^2\leqslant\int\limits_0^L\zeta^T\mathcal{P}\zeta+\frac{1}{2\gamma}\|\delta_1(.,t)\|^2+\frac{k_1}{2}\delta_2^2.
\end{align}
where $\mathcal{P}$ is chosen such that the maximum singular value of $\mathcal{P}$ is greater than the maximum singular value of $\mathcal{P}_i, \forall i.$ i.e. $\overline{\sigma}(\mathcal{P}_i)\leqslant \overline{\sigma}(\mathcal{P}), \, \forall i.$ Next, we can choose a constant $\lambda_s$ such that it satisfies the LMI $(\mathcal{P}+\lambda_s\mathbf{I})<0$. This in turn yields
\begin{align}\label{final_edot}
     \dot{\mathcal{E}}(t)\leqslant -\lambda_s \mathcal{E}(t)+\frac{1}{2\gamma}\|\delta_1(.,t)\|^2+ \frac{k_1}{2}\delta_2^2(t).
\end{align}
Consequently, using Gr\"{o}nwall's inequality we can write \cite{grimshaw1991nonlinear}
\begin{align}
    \mathcal{E}(t)\leqslant \mathcal{E}(0)e^{-\lambda_s t}+m_1\sup_t\left( \|\delta_1(.,t)\|^2+\delta_2^2(t)\right),
\end{align}
where $m_1 = \max(1,k_1\gamma)/(2\gamma\lambda_s)$.

Next, we attempt to prove two inequalities: (I) $m_2 r_s(t) \leqslant \mathcal{E}(t)$  and (II)  $\mathcal{E}(0)\leqslant m_3 r_s(0)$ for $m_2, m_3>0$. Using Young's Inequality one more time in \eqref{residual_social} we obtain
\begin{align}
    r_s(t) \leqslant m_4\sum_i [\widetilde{W}^2(z_i,t)+\widetilde{V}^2(z_i,t)],
\end{align}
for some $m_4=\max\{(1+\frac{1}{m_5}),(1+m_5)\}$ and $m_5>0$. Subsequently, it is trivial to note that $\sum_i[\widetilde{W}^2(z_i,t)+\widetilde{V}^2(z_i,t)]\leqslant 2 e^L \mathcal{E}(t)$. This yields our required inequality (I) for $m_2 = e^{-L}/(2m_4)$. 

To prove the next inequality (II), we consider bounded initial conditions for the error system i.e. $0<|\widetilde{W}(x,0)|, |\widetilde{V}(x,0)|< \infty$. Using this assumption we can obtain (II) where $m_3 = \frac{L\max_x\{(\left|\widetilde{W}(x,0)\right|+\left|\widetilde{V}(x,0)\right|)^2\}}{\min_x\{\widetilde{W}^2(x,0),\widetilde{V}^2(x,0)\}}$. Finally, choosing $c_1 = m_3/m_2$ and $c_2 = m_1/m_2$, we obtain \eqref{thm12}. Moreover, when there is no attack i.e. $\delta_1\equiv 0$ and $\delta_2\equiv 0$, we can obtain \eqref{thm11}.
\end{proof}

\section{Physical Data-Based Filter}

Unlike the Social Data-based Filter, the Physical Data-based Filter obtains its measurement from the outlet of the traffic (at $x=L$) in the form of flow measurement and is given by
\begin{align}
&y_p(t) = v(L,t).
\end{align}

As before, using the system model \eqref{sys11}, \eqref{bc11}-\eqref{bc22} and output injection, we can write the model for the Physical Data-based Filter to be:
\begin{align}\label{physical1}
&\hat{w}_t=-k_1\hat {w}_x-k_2{\hat {w}}+\gamma_1(x)(y_p-\hat {y}_p),\\
&\hat {v}_t=-k_3\hat {v}_x-k_2{\hat {w}}+ \gamma_2(x)(y_p-\hat {y}_p),\\\label{physical2}
&\hat{w}(0,t) =v(0,t), \quad \hat{w}(L,t)=\hat{v}(L,t),\\
&\hat {y}_p=\hat{v}(L,t),
\end{align}
where $\hat{w}, \hat{v}$ and $\hat{y}_p$ represent the estimates of $w, v$ and $y_p$ respectively obtained from the Physical Data-based Filter, and $\gamma_1(x)$ and $\gamma_2(x)$ are filter gains. The errors for the Physical Data-based Filter is then defined as  $\widehat{W}:=w-\widehat{w}$ and $\widehat{V}:=v-\widehat{v}$ and their dynamics is given by the difference of \eqref{physical1}-\eqref{physical2} and \eqref{sys11},\eqref{bc11}-\eqref{bc22}:
\begin{align}\label{phy_err1}
&\widehat{W}_t=-k_1\widehat{W}_x-k_2{\widehat{W}}-\gamma_1(x)(y_p-\hat {y}_p),\\\label{phy_err2}
&\widehat{V}_t=-k_3\widehat{V}_x-k_2{\widehat{W}}-\gamma_2(x)(y_p-\hat {y}_p)+\delta_1,\\
&\widehat{W}(0,t)=\delta_2, \quad \widehat{W}(L,t)=\widehat{V}(L,t)\label{phy_err3}.
\end{align}

The residual for the Physical Data-based Filter is given by
\begin{align}
    \label{residual_physical}
& r_p(t) = \widehat{V}^2(L,t)\diff x.
\end{align}
In order to show the asymptotic stability (under no attack) and input-to-state stability (under attack) of this physical data based residual signal $r_p(t)$, we transform the  error residual dynamics  \eqref{phy_err1}-\eqref{phy_err3} using backstepping transformation \cite{roy_dey2020}. The backstepping transformations used here are given as follows:
\begin{align}\label{bst1}
    \widehat{W}(x,t) = \phi(x,t) &- \int_x^L  \mathcal{F}(x,y)\phi(y,t)\diff y,\\ \label{bst2}
    \widehat{V}(x,t) = \psi(x,t) &-  \int_x^L  \mathcal{G}(x,y)\psi(y,t)\diff y-\int_x^L \mathcal{H}(x,y)\phi(y,t)\diff y.
\end{align}
These transformation maps the coupled PDE \eqref{phy_err1}-\eqref{phy_err3} to a boundary-condition coupled target PDE given as follows:
\begin{align}\label{target}
    &\phi_t = -k_1 \phi_x-k_2\phi,\quad \psi_t = -k_3 \psi_x + \widetilde{\delta_1},\\\label{target_BC}
    &\phi(L,t) = \psi(L,t),\quad \phi(0,t) ={\delta_2},
\end{align}
where the $\widetilde{\delta_1}$ is connected to the in-domain attack function $\delta_1$ in the following way:
\begin{align}
    &\delta_1(x,t)  = \widetilde{\delta_1}(x,t) -  \int_x^L  \mathcal{F}(x,y)\widetilde{\delta_1} (y,t)\diff y. 
\end{align}

Using the backstepping transformation \eqref{bst1}-\eqref{bst2} and comparing the target dynamics \eqref{target}-\eqref{target_BC} with the residual dynamics \eqref{phy_err1}-\eqref{phy_err3}, we obtain the dynamics of the kernels of the backstepping transformation:
\begin{align}
    &\mathcal{F}_x + \mathcal{F}_y = 0,\quad \mathcal{F}(0,y)=0,\\
    &\mathcal{G}_x + \mathcal{G}_y = 0,\quad  \mathcal{G}(0,y)=0,\\
    & k_3\mathcal{H}_x + k_1\mathcal{H}_y -k_2\mathcal{H} = k_2 \mathcal{F},\quad\mathcal{H}(0,y)=0 .
\end{align}

Furthermore, the physical filter gains $\gamma_1(x)$ and $\gamma_2(x)$ must be chosen so that
\begin{align}\label{gamma_1}
   \gamma_1(x) = -\frac{\mathcal{F}(x,L)\kappa_1 }{\rho^*},&\quad \gamma_2(x) = -\frac{\mathcal{G}(x,L)\kappa_3 }{\rho^*}.
\end{align}
Moreover, since $\widehat{V}(0,t)=0$ it yields
\begin{align} \label{psi_0}
    \psi(0,t) =  0.
\end{align}

With this we present our second theorem that specify the design criteria for the Physical Data-based Filter.

\begin{thmm}[Convergence of Residual Dynamics for Physical Data-based Filter]
Consider the error dynamics \eqref{phy_err1}-\eqref{phy_err3} and residual definition \eqref{residual_physical}. The residual signal is asymptotically stable in the following sense:
\begin{align}\label{thm21}
    & r_p(t)\leqslant c_3 r_p(0)\exp{(-\lambda_p t)}, 
\end{align}
 for some $c_3>0$ and $\lambda_p >0$ without any attack. Furthermore, the residual signal is input-to-state state stability in the following sense:
 \begin{align}\label{thm22}
    & r_p(t)\leqslant c_3 r_p(0)\exp{(-\lambda_p t)}+c_4\sup_t\left( \|\delta_1(.,t)\|^2+\delta_2^2(t)\right), 
\end{align}
 for some $c_4>0$ under attack. These conditions are satisfied if the filter gains meet the prescribed conditions in \eqref{gamma_1}.
 
\end{thmm}

\begin{proof}
The theorem can be proved following an approach similar to the one presented in \cite{roy_dey2020}.
\end{proof}

\section{Comparator}
As shown in Fig. \ref{SDF}, the residuals generated by the Social Data-based Filter $r_s(t)$ and Physical Data-based Filter $r_p(t)$ are propagated to the Comparator block. The purpose of the Comparator is to process generated residuals in a meaningful way in order to provide an attack-detection decision.

From the previous theoretical analysis, it is evident that under no attack scenarios both the residuals converge to zero. However, in practice, even under no attack  these residuals will never be identically zero because of system uncertainties arising from model and measurement inaccuracies. Accordingly, we define a threshold such that residual signals greater than that pre-determined value would indicate presence of an attack. Such computation of threshold is standard in the community of fault detection and generally obtained under no attack scenarios \cite{mansouri2016statistical}. In this work, we chose the threshold by finding the maximum value of the residual under no attack conditions \cite{Roy_its_2021}.

Once these thresholds are obtained for both the Physical and Social Data-based Filters, both the residuals are compared to their corresponding thresholds and a logical output is generated. In case of the Physical Data-based Filter, if the threshold is given by $r_{th,p}$, then $r_p(t)\geqslant r_{th,p}$ produces a \textit{high} and a \textit{low} otherwise. Similarly, for the Social Data-based Filter with a threshold $r_{th,s}$, if $r_s(t)\geqslant r_{th,s}$ then it produces a \textit{high}, otherwise a \textit{low}. 

\begin{table}[t]
\centering
\caption{Attack Detection Logic.}
\label{table1_1}
\begin{tabular}{ccc}
  \hline\hline
  Physical Residual  & Social Residual& Decision\\
  \hline
  Low $\left[r_p(t)<r_{th,p}\right]$ & Low $\left[r_s(t)<r_{th,s}\right]$ & No \\\hline
  High $\left[r_p(t)\geqslant r_{th,p}\right]$ & Low $\left[r_s(t)<r_{th,s}\right]$ & Yes \\\hline
  Low $\left[r_p(t)<r_{th,p}\right]$ & High $\left[r_s(t)\geqslant r_{th,s}\right]$ & Yes \\\hline
  High $\left[r_p(t)\geqslant r_{th,p}\right]$ & High $\left[r_s(t)\geqslant r_{th,s}\right]$ & Yes \\
\hline
\end{tabular}
\renewcommand{\arraystretch}{1}
\end{table}

Next, these logical signals are compared to determine the final attack-detection in the following way: if at least one of the residuals produces a high, then a positive attack decision is confirmed. In other words, only if both the filters produce a low signal, a negative attack decision is made. This implies that proposed cyber-attack detection schemes utilizes the redundancies of both physical data and social data to increase effective attack detection.  The complete attack-detection logic  for the Comparator block is shown in Table \ref{table1_1}.

\section{Simulation results}

In this section, we present simulation results to show the effectiveness of our cyber-attack detection scheme using the redundancies in Social and Physical Data-based Filters. We first train the LSTM neural network-based fake data isolator and test its performance over the testing dataset. Thereafter, we simulate both the microscopic model \eqref{micro} and the equivalent macroscopic model \eqref{lin1}-\eqref{lin2} under equivalent operating conditions. The microscopic model is used to generate vehicle-level data for \textit{Social Data-based Filter} whereas the macroscopic model is used to generate physical measurement data for the \textit{Physical Data-based Filter}. In order to emulate realistic scenario, we have injected noise in both physical and social measurements. The source of noise in physical data is the inaccuracies in loop detector measurements whereas the noise in social data arises from intermittent and delayed nature of GPS signals and tweets. The prescribed thresholds for the Social and Physical Data-based Filters are chosen as $1\times 10^{-5}$ and $2\times 10^{-5}$, respectively.


\subsection{Performance of LSTM for fake data isolation}

We  split  the  dataset described in Section III A into  35\%  testing  and   65\% training  datasets to train our LSTM neural network for fake data isolation/classification,. The neural network model is trained for 20 epochs and the model accuracy and training loss improvement over the epochs are shown in Fig. \ref{fig:LSTM_a}.


This classification model is tested using the multiple metrics. The confusion matrix for the test dataset is shown in Fig. \ref{fig:confusion_matrix}. First, we evaluate the performance of the network using the definition of 
$
    \text{accuracy} = \frac{\text{TP+TN}}{\text{TP+FP+TN+FN}},
$
where TP, FP, TN and FN imply True Positive, False Positive, True Negative and False Negative, respectively.  The accuracy of our network on the test dataset is 89\% leading to a misclassification rate of 11\%. 

As the error costs of  positive and negative classification are different, we look at the Sensitivity metric of the trained network which are defined as follows:
$\text{Sensitivity} = \frac{\text{TP}}{\text{TP+FN}}$. 
We intended to build a highly sensitive network such that maximum number of false social data are flagged. The Sensitivity of our network was obtained to be 98\%.


\begin{figure}[!b]
\centering
\subfloat[]{\includegraphics[width=0.3\columnwidth]{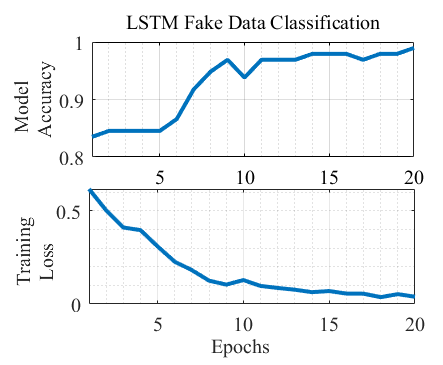}%
\label{fig:LSTM_a}}
\subfloat[]{\includegraphics[width=0.25\columnwidth]{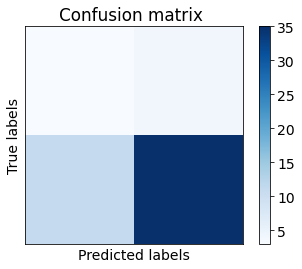}%
\label{fig:confusion_matrix}}
\caption{(a) Model accuracy and training loss for LSTM-based fake social data classifier. (b) Confusion matrix for detecting fake social data.}
\label{fig:Lstm_results}
\end{figure}

\subsection{Performance of detection scheme under nominal condition}
The macroscopic velocity and density profile for the traffic under nominal (no attack) condition (given by \eqref{lin1}-\eqref{lin2}) is shown in {Fig. \ref{res}}.  The social and physical residuals generated by Social and Physical Data-based Filters are shown in {Fig. \ref{res}}. We observe that both the residuals converge close to zero starting from non-zero initial conditions. Note that the residuals do not converge exactly to zero due to the presence of noise. Nevertheless, under this nominal operating condition, both the residuals remain under their respective threshold values (shown by dashed line in Fig. \ref{res}).


\begin{figure}[ht!]
    \centering
    \includegraphics[scale=0.2,width=0.5\textwidth]{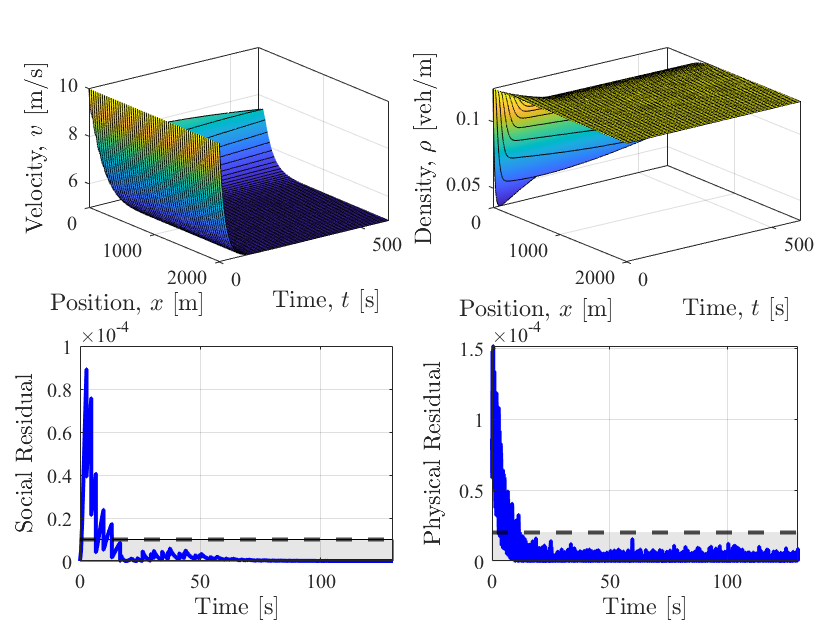}
    \caption{(Top) Velocity and density profile of freeway traffic under nominal operating condition. (Bottom) Social and physical data based residuals under nominal operating conditions.}
    \label{res}
\end{figure}


\subsection{Performance of detection scheme under in-domain cyber-attacks}
In this study, we illustrate the advantage of using both social and physical data as opposed to just physical data. We show three cases demonstrating the advantage of our proposed method.

\textit{Case I:} We inject a ``stealthy" in-domain attack at $100s$  to the velocity profile. This attack is ``stealthy'' in the sense that it is does not show up at the physical data sensor of the system (that is, in $\rho(L)$) \cite{roy_dey2020}. Essentially, this attack acts as an high amplitude perturbation occurring somewhere in-domain. The velocity and density response under this attack is shown in {Fig. \ref{fig:attack_phy_low} } where it is evident that the effect of in-domain attack does not show up significantly at the outlet measurement. Fig. \ref{fig:attack_phy_low} also shows the residual response under attack where the residual for both the filters remain below threshold until the injection of attack at 100$s$. However, since the outlet measurement is used by the Physical Data-based Filter, the residual of the Physical Data-based Filter remains below its threshold  even after the attack injection. On the other hand, as Social Data-based Filter uses in-domain data from the vehicles, its residual cross its threshold within 2.5$s$ of the attack. This provides a \textit{high} to the comparator which can make a positive attack detection decision using {Table \ref{table1_1}}. This implies that such attacks would remain undetected if Physical Data-based Filter is used exclusively. However, they will be detected by our scheme as we exploit the redundancies between social and physical data.

\begin{figure}[ht!]
    \centering
    \includegraphics[scale=0.2,width=0.5\textwidth]{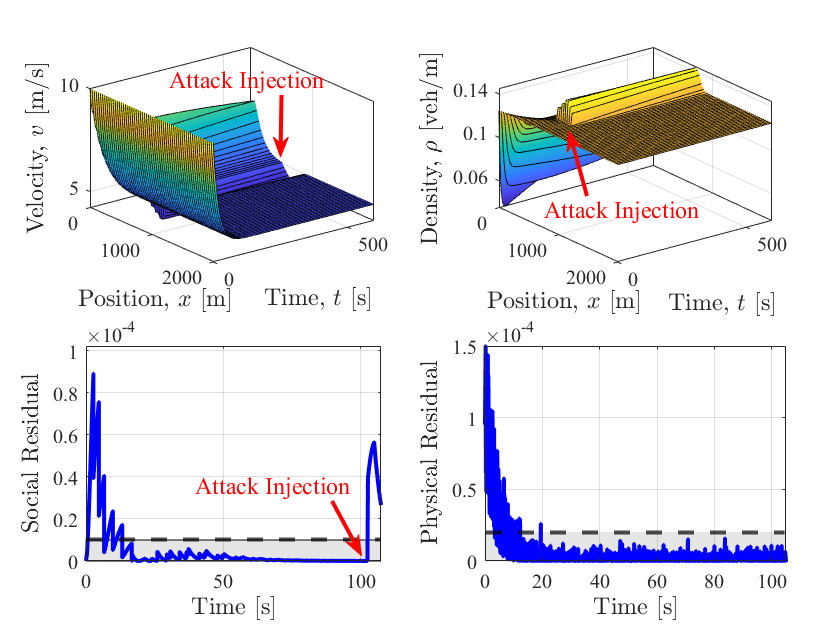}
    \caption{(Top) Velocity and density profile of freeway traffic with an in-domain attack. (Bottom) Social and physical data based residuals under cyber-attack. Attack detected by the Social Data-based Filter only. }
    \label{fig:attack_phy_low}
\end{figure}


It should also be noted here that Social Data-based Filters can only detect an attack if the social sensors are in the \textit{spatio-temporal impact zone} of the attack. Strictly speaking, this implies that the social sensor must lie in the cone of influence of the attack propagation. This indicates that the location and timing of the social sensor in relation to the injected attack is a key point in this setting.

\textit{Case II:} Social data is non-stationary as well as intermittent. This implies that under certain scenarios, social data sensors might not be available near the spatio-temporal impact zone of an attack.  In this case, if a physical sensor is present in the zone, Physical Data-based Filter can detect these attacks while the Social Data-based filter cannot. For example, in our setting, if the in-domain attack is injected closer to the physical sensor at the outlet at 100$s$. This is evident from the traffic velocity and density profile in { Fig. \ref{fig:attack_phy_high}}. We can also observe from {Fig. \ref{fig:attack_phy_high}} that the residuals of both the filters remain within the thresholds until the injection of attack at 100$s$. After the attack injection, the residual of the Physical Data-based Filter crosses the threshold in 1$s$ providing a \textit{high} to the comparator. On the other hand, since no vehicles are present in the immediate impact zone to transmit social data, the residual for the Social Data-based Filter remains below the threshold, signaling a \textit{low}. Using these two signals, the comparator can decide an attack has occurred in the system using Table \ref{table1_1}. 

It is to be noted that physical sensors are fixed and do not have the maneuverability of social data sensors. Moreover, they are expensive to install and cannot be deployed in large numbers. However, we see from this case study that despite such disadvantages of physical sensors, they can provide valuable attack information in case social sensors are unavailable. Thus, a faster response and mitigation can be undertaken using the output of the Physical Data-based Filter as well.

\begin{figure}[ht!]
    \centering
    \includegraphics[scale=0.1,width=0.5\textwidth]{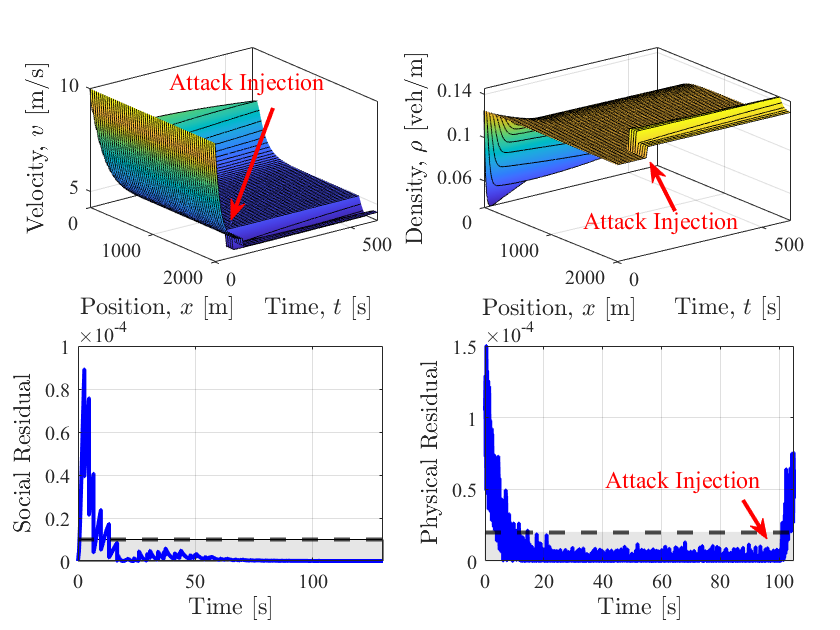}
    \caption{(Top) Velocity and density profile of freeway traffic with an in-domain attack near the outlet. (Bottom) Social and physical data based residuals under cyber-attack. Attack detected by the Physical Data-based Filter only. }
    \label{fig:attack_phy_high}
\end{figure}

\textit{Case III:} Finally, we present the case where the residuals of both the filters are affected by a cyber-attack. In this scenario, a wide-spread in-domain attack is injected the effect of which can be seen from the velocity and density profiles of the traffic in Fig. \ref{fig:attack_phy_high_soc_high}. Such a case is an ideal scenario where there are  social data sensors as well as physical data sensor in the impact zone of the attack and  both the residuals cross their thresholds after the attack injection. Consequently, the comparator receives \textit{high} signals from both the filters and is able to make a positive attack detection decision.  We also note here that, similar to previous cases, the residuals of both the filters remain below their respective thresholds before the injection of attack at 100$s$. Notably, this scenario can occur under two possible situations: (i) the magnitude of the attack is large such that it has a wide spatio-temporal impact zone, or (ii) irrespective of the size of the attack, both physical as well as social sensors are present in the zone of impact of the attack.

\begin{figure}[ht!]
    \centering
    \includegraphics[scale=0.2,width=0.5\textwidth]{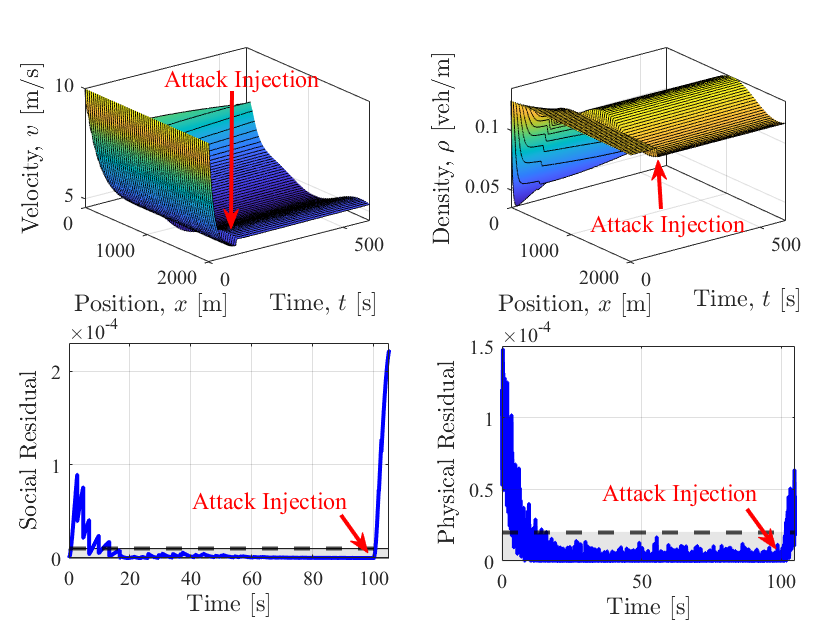}
    \caption{(Top) Velocity and density profile of freeway traffic with a wide-spread in-domain attack. (Bottom) Social and physical data based residuals under cyber-attack. Attack detected by both Physical and Social Data-based Filters. }
    \label{fig:attack_phy_high_soc_high}
\end{figure}


\section{Conclusions}
In this paper, we explore a cyber-attack detection scheme for socio-technical transportation systems. We consider the availability of both physical and social data from vehicle/infrastructure sensors and human users' mobile devices, respectively. We exploit the redundancy between these social and physical data and design our cyber-attack detection scheme based on macroscopic traffic model. Essentially, the scheme consists of a Social Data-based Filter and a Physical Data-based Filter running in parallel and producing residuals. The attack decision is made by comparing these two residuals. We have analyzed the mathematical properties of such filters using Lyapunov's stability theory. Furthermore, we performed simulation studies that illustrate the efficacy of the proposed scheme. As future studies, we plan to explore the effectiveness of the proposed scheme  under various types of uncertainties related to social data generation and processing.

\bibliography{ref1}

\begin{thebibliography}{10}

\bibitem{gowrishankar2014including}
S.~Gowrishankar, E.~Stern, and D.~B. Work, ``Including the social component in
  smart transportation systems,'' in {\em National Workshop on Transportation
  Cyber-Physical Systems}, 2014.

\bibitem{zheng2018framework}
Z.~Zheng, C.~Wang, P.~Wang, Y.~Xiong, F.~Zhang, and Y.~Lv, ``Framework for
  fusing traffic information from social and physical transportation data,''
  {\em PloS one}, vol.~13, no.~8, p.~e0201531, 2018.

\bibitem{becker2000sensor}
J.~C. Becker and A.~Simon, ``Sensor and navigation data fusion for an
  autonomous vehicle,'' in {\em Proceedings of the IEEE Intelligent Vehicles
  Symposium 2000 (Cat. No. 00TH8511)}, pp.~156--161, IEEE, 2000.

\bibitem{grid_data_fusion}
M.~Kordestani and M.~Saif, ``Data fusion for fault diagnosis in smart grid
  power systems,'' in {\em 2017 IEEE 30th Canadian Conference on Electrical and
  Computer Engineering (CCECE)}, pp.~1--6, IEEE, 2017.

\bibitem{defense_data_fusion}
C.~Harris, A.~Bailey, and T.~Dodd, ``Multi-sensor data fusion in defence and
  aerospace,'' {\em The Aeronautical Journal}, vol.~102, no.~1015,
  pp.~229--244, 1998.

\bibitem{book_data_fusion}
D.~P. Mandic, D.~Obradovic, A.~Kuh, T.~Adali, U.~Trutschell, M.~Golz,
  P.~De~Wilde, J.~Barria, A.~Constantinides, and J.~Chambers, ``Data fusion for
  modern engineering applications: An overview,'' in {\em International
  Conference on Artificial Neural Networks}, pp.~715--721, Springer, 2005.

\bibitem{Data_fusion}
H.~F. Durrant-Whyte, ``Sensor models and multisensor integration,'' {\em The
  International Journal of Robotics Research}, vol.~7, no.~6, pp.~97--113,
  1988.

\bibitem{zheng2016big}
X.~Zheng, W.~Chen, P.~Wang, D.~Shen, S.~Chen, X.~Wang, Q.~Zhang, and L.~Yang,
  ``Big data for social transportation,'' {\em IEEE Transactions on Intelligent
  Transportation Systems}, vol.~17, no.~3, pp.~620--630, 2016.

\bibitem{Roy_its_2021}
T.~{Roy}, A.~{Tariq}, and S.~{Dey}, ``A socio-technical approach for resilient
  connected transportation systems in smart cities,'' {\em IEEE Transactions on
  Intelligent Transportation Systems}, pp.~1--10, 2021.

\bibitem{biron2017resilient}
Z.~A. Biron, S.~Dey, and P.~Pisu, ``Resilient control strategy under denial of
  service in connected vehicles,'' in {\em 2017 American Control Conference
  (ACC)}, pp.~4971--4976, IEEE, 2017.

\bibitem{he2012mitigating}
L.~He and W.~T. Zhu, ``Mitigating dos attacks against signature-based
  authentication in vanets,'' in {\em 2012 IEEE International Conference on
  Computer Science and Automation Engineering (CSAE)}, vol.~3, pp.~261--265,
  IEEE, 2012.

\bibitem{petrillo2018collaborative}
A.~Petrillo, A.~Pescap{\'e}, and S.~Santini, ``A collaborative approach for
  improving the security of vehicular scenarios: The case of platooning,'' {\em
  Computer Communications}, vol.~122, pp.~59--75, 2018.

\bibitem{merco2018replay}
R.~Merco, Z.~A. Biron, and P.~Pisu, ``Replay attack detection in a platoon of
  connected vehicles with cooperative adaptive cruise control,'' in {\em 2018
  Annual American Control Conference (ACC)}, pp.~5582--5587, IEEE, 2018.

\bibitem{REILLY2016}
J.~Reilly, S.~Martin, M.~Payer, and A.~M. Bayen, ``Creating complex congestion
  patterns via multi-objective optimal freeway traffic control with application
  to cyber-security,'' {\em Transportation Research Part B: Methodological},
  vol.~91, pp.~366 -- 382, 2016.

\bibitem{roy_dey2020}
T.~{Roy} and S.~{Dey}, ``Secure traffic networks in smart cities: Analysis and
  design of cyber-attack detection algorithms,'' in {\em 2020 American Control
  Conference (ACC)}, pp.~4102--4107, 2020.

\bibitem{jahanshahi2018attack}
N.~Jahanshahi and R.~M. Ferrari, ``Attack detection and estimation in
  cooperative vehicles platoons: A sliding mode observer approach,'' {\em
  IFAC-PapersOnLine}, vol.~51, no.~23, pp.~212--217, 2018.

\bibitem{dadras2018identification}
S.~Dadras, S.~Dadras, and C.~Winstead, ``Identification of the attacker in
  cyber-physical systems with an application to vehicular platooning in
  adversarial environment,'' in {\em 2018 Annual American Control Conference
  (ACC)}, pp.~5560--5567, IEEE, 2018.

\bibitem{paule2019social}
J.~D.~G. Paule, Y.~Sun, and Y.~Moshfeghi, ``On fine-grained geolocalisation of
  tweets and real-time traffic incident detection,'' {\em Information
  Processing \& Management}, vol.~56, no.~3, pp.~1119--1132, 2019.

\bibitem{andrea2015social}
E.~D'Andrea, P.~Ducange, B.~Lazzerini, and F.~Marcelloni, ``Real-time detection
  of traffic from twitter stream analysis,'' {\em IEEE transactions on
  intelligent transportation systems}, vol.~16, no.~4, pp.~2269--2283, 2015.

\bibitem{herrera2010probe}
J.~C. Herrera and A.~M. Bayen, ``Incorporation of lagrangian measurements in
  freeway traffic state estimation,'' {\em Transportation Research Part B:
  Methodological}, vol.~44, no.~4, pp.~460--481, 2010.

\bibitem{Barreau2020probe}
M.~{Barreau}, A.~{Selivanov}, and K.~H. {Johansson}, ``Dynamic traffic
  reconstruction using probe vehicles,'' in {\em 2020 59th IEEE Conference on
  Decision and Control (CDC)}, pp.~233--238, 2020.

\bibitem{canepa_2013}
E.~S. {Canepa} and C.~G. {Claudel}, ``Spoofing cyber attack detection in
  probe-based traffic monitoring systems using mixed integer linear
  programming,'' in {\em 2013 International Conference on Computing, Networking
  and Communications (ICNC)}, pp.~327--333, 2013.

\bibitem{daganzo1995requiem}
C.~F. Daganzo, ``Requiem for second-order fluid approximations of traffic
  flow,'' {\em Transportation Research Part B: Methodological}, vol.~29, no.~4,
  pp.~277--286, 1995.

\bibitem{AwRascle}
A.~Aw and M.~Rascle, ``Resurrection of "second order" models of traffic flow,''
  {\em SIAM Journal on Applied Mathematics}, vol.~60, no.~3, pp.~916--938,
  2000.

\bibitem{costeseque2015lax}
G.~Costeseque, ``Lax-hopf formula for arz traffic flow model,'' 2015.

\bibitem{paule2019fine}
J.~D.~G. Paule, Y.~Sun, and Y.~Moshfeghi, ``On fine-grained geolocalisation of
  tweets and real-time traffic incident detection,'' {\em Information
  Processing \& Management}, vol.~56, no.~3, pp.~1119--1132, 2019.

\bibitem{LSTM}
S.~Hochreiter and J.~Schmidhuber, ``Long short-term memory,'' {\em Neural
  Computation}, vol.~9, pp.~1735--1780, 1997.

\bibitem{Coulter_fake_tweets}
R.~{Coulter}, Q.~{Han}, L.~{Pan}, J.~{Zhang}, and Y.~{Xiang}, ``Data-driven
  cyber security in perspective—intelligent traffic analysis,'' {\em IEEE
  Transactions on Cybernetics}, vol.~50, no.~7, pp.~3081--3093, 2020.

\bibitem{Wang_distributed}
J.~{Wang}, Y.~{Liu}, and C.~{Sun}, ``Luenberger observer design for state
  estimation of a linear parabolic distributed parameter system with discrete
  measurement sensors,'' in {\em 2016 12th World Congress on Intelligent
  Control and Automation (WCICA)}, pp.~1123--1128, 2016.

\bibitem{collocated_lyapunov_proof}
J.-W. Wang and H.-N. Wu, ``Lyapunov-based design of locally collocated
  controllers for semi-linear parabolic pde systems,'' {\em Journal of the
  Franklin Institute}, vol.~351, no.~1, pp.~429--441, 2014.

\bibitem{grimshaw1991nonlinear}
R.~Grimshaw, {\em Nonlinear ordinary differential equations}, vol.~2.
\newblock CRC Press, 1991.

\bibitem{mansouri2016statistical}
M.~Mansouri, M.~Sheriff, R.~Baklouti, M.~Nounou, H.~Nounou, A.~B. Hamida, and
  N.~Karim, ``Statistical fault detection of chemical process-comparative
  studies,'' {\em Journal of Chemical Engineering \& Process Technology},
  vol.~7, no.~1, pp.~282--291, 2016.

\end{thebibliography}

%








\end{document}